\documentclass[11pt]{article}

\usepackage{amscd}
\usepackage{amsfonts}
\usepackage{amsmath}
\usepackage{amsthm}
\usepackage{hyperref}
\usepackage{bm}

\allowdisplaybreaks

\theoremstyle{plain}
\newtheorem{proposition}{Proposition}
\newtheorem{corollary}{Corollary}
\newtheorem{The}{Theorem}

\theoremstyle{definition}
\newtheorem{definition}{Definition}[section]
\newtheorem{example}{Example}[section]

\theoremstyle{remark}

\def\fc{{\mathfrak{C}}}
\def\fv{{\mathfrak{V}}}
\def\fr{{\mathfrak{R}}}
\def\bbbn{{\mathbb N}}
\def\bbbc{{\mathbb C}}

\def\cD{{\mathcal D}}
\def\dm{{\bullet}}
\def\v{{\bf V}}

\title{Symbolic Representation and Classification of $N=1$ Supersymmetric Evolutionary Equations}

\author{{\sc Kai Tian} \\
Department of Mathematics \\
China University of Mining \& Technology, Beijing, P. R. China  \\
and \\
{\sc Jing Ping Wang} \\
School of Mathematics, Statistics \& Actuarial Science \\
University of Kent, Canterbury, UK}

\date{}

\begin{document}
\maketitle

\begin{abstract}
We extend the symbolic representation to the ring of $N=1$ supersymmetric differential polynomials, 
and demonstrate that operations on the ring, such as the super derivative, Fr\'{e}chet derivative and super commutator, can be carried out in the symbolic way. 
Using the symbolic representation, we classify scalar $\lambda$-homogeneous $N=1$ supersymmetric evolutionary equations with 
nonzero linear term when $\lambda>0$ for arbitrary order and give a comprehensive description of all such integrable equations.
\end{abstract}

\section{Introduction}
Classification of all integrable equations for a given family of nonlinear 
equations is one of fundamental problems in the field of soliton theory, and 
can be tackled by many methods, among which the symmetry approach has been 
proved to be very efficient and powerful. The symmetry approach takes the 
existence of infinitely many higher order infinitesimal symmetries as the 
definition of integrability. It has been used to classify large classes of both 
nonlinear partial differential equations and difference-differential equations. 
We refer the reader to the recent review papers \cite{asy00,mikh09} and 
references therein for details. 

The ultimate goal is to obtain the classification result for integrable 
equations of any order, i.e. the global classification result \cite{mikh09s}. So 
far, results of this kind have been obtained for scalar homogeneous 
evolutionary equations in both commutative \cite{sand98} and 
non-commutative \cite{olver00} cases, and a class of second order (in
time) non-evolutionary equations of odd order in spatial variable \cite{nw07}. 
They have been achieved mainly due to the symbolic representation of the ring of 
differential polynomials. The symbolic representation (an abbreviated form of the 
Fourier transformation), originally applied to the theory of integrable 
equations by Gel'fand and Dikii \cite{gd75}, enables us to translate 
the question of integrability into the problem of divisibility of certain 
multi-variable polynomials, which can be solved by results from number theory and 
algebraic geometry. It is a suitable tool to study 
integrability of non-commutative \cite{olver00}, non-evolutionary \cite{nw07,mikh07}, 
non-local (integro-differential) \cite{mikh03}, multi-component \cite{sand04} 
and multi-dimensional equations \cite{wang06}. 

Supersymmetry was introduced into the 
soliton theory in the late of 1970s. Many classical integrable equations have 
been successfully embedded in their supersymmetric extensions, including the 
supersymmetric sine-Gordon equation \cite{kulish78}, the supersymmetric 
Korteweg-de Vries (KdV) equation\cite{manin85,mathieu88}, the supersymmetric 
nonlinear Schr\"{o}dinger equations \cite{roelofs92}, the supersymmetric two 
boson equation \cite{brunelli94}, the supersymmetric Harry Dym equations 
\cite{brunelli03}, etc. These supersymmetric integrable prototypes have been 
extensively studied in the past four decades, and shown to have various novel 
features, such as fermionic nonlocal conserved densities \cite{kersten88}, 
non-unique roots of Lax operator \cite{oevel91} and odd Hamiltonian structures 
\cite{popo09}. Recently, discrete integrable systems on Grassmann algebras \cite{gramik13,xue13,xue14,xue15}, 
as well as Grassmann extensions of Yang-Baxter maps \cite{gramik16}, have been established 
due to the development of Darboux-B\"{a}cklund transformations of super and/or
supersymmetric integrable equations.

Though the list of supersymmetric integrable equations are long, 
it may be far from complete taking account of the fact that the (non-)existence of supersymmetric 
counterparts of some classical models, like the Ibragimov-Shabat equation 
\cite{ibra80} and the Kaup-Kupershmidt equation \cite{kaup80}, is not 
proved. Some efforts have been made to enrich the garden of supersymmetric
integrable equations. Fermionic extensions of Burgers and Boussinesq equations 
were constructed and shown to have infinitely many symmetries \cite{kiswolf06}. 
A class of fifth order supersymmetric evolutionary equations (which are $3/2$-homogeneous in our terminology.) admitting seventh
order symmetries were presented, and all of them are claimed to be integrable \cite{tian10}. 

In this paper, we develop the symbolic representation for the ring of $N=1$ 
supersymmetric differential polynomials so that we can globally classify scalar 
$\lambda$-homogeneous evolutionary equation of the form
\begin{equation}\label{susymodel}
\Phi_t=\Phi_n+ K\Big(\Phi, (\cD \Phi), \Phi_1, (\cD \Phi_1), \cdots, 
\Phi_{n-1}, (\cD \Phi_{n-1})\Big),
\end{equation}
where the super function $\Phi = \Phi(x,\theta,t)$ is fermionic 
(anti-commutative), $\Phi_k$ stands for the $k$-th order derivative of $\Phi$ 
with respect to $x$, and $\cD$ denotes the super derivative. 
(Detailed explanations about our notations will be given in Section \ref{sec:2}.) 
We show that there are only eight hierarchies of  
nontrivial supersymmetric integrable equations. To the best of our knowledge, the hierarchy of 
the supersymmetric fifth order modified KdV equation \eqref{s5mkdv} has not appeared in the literature. 
This classification result also enables us to give definite answers to the non-existence of the above equations.

The arrangement of this paper is as follows. In Section \ref{sec:2}, we give a basic 
introduction on a super Lie bracket for the super differential ring of smooth 
functions in $\Phi$ and its (super) derivatives, and then define the concepts of 
infinitesimal symmetries and integrability for $N=1$
supersymmetric evolutionary equations. In Section \ref{sec:3}, we present 
the main result of this paper, i.e. a complete list of nontrivial $\lambda$-homogeneous 
($\lambda>0$) supersymmetric integrable equations.
This result is global in the sense that if a $\lambda$-homogeneous $N=1$ 
supersymmetric evolutionary equation of arbitrary order is nontrivial and integrable, 
then it must belong to one of symmetry hierarchies of eight equations we identified.
The trivial supersymmetric integrable equations of such kind are listed in Appendix \ref{app:a}.
In Section \ref{sec:4}, we develop the symbolic representation converting every $N=1$ supersymmetric 
differential polynomial to a unique multi-variable polynomial (and vice versa), which is anti-symmetric 
with respect to symbols of fermionic variables and meanwhile symmetric with respect to symbols of bosonic variables.
Inferred from this one-to-one correspondence, we show that operations of supersymmetric differential 
polynomials, such as partial derivatives, the super derivative, derivative with respect to $x$
Fr\'{e}chet derivative and super commutator, can be carried out in the symbolic way. 
Interestingly and importantly, the super commutator of the linear term $\Phi_k$ with 
any supersymmetric differential polynomial has the same symbolic formulation as 
its counterpart in the classical case \cite{sand98}. This remarkable fact guarantees 
that equation \eqref{susymodel} can be globally classified via the same strategy.
We explain it in Section \ref{sec:5} although some detailed proofs are skipped.
Concluding remarks and discussions are given in the last section.

\section{Symmetries of $N=1$ supersymmetric evolutionary equations}\label{sec:2}
In this section, we give a brief account on symmetries for $N=1$ supersymmetric evolutionary equations. Algebraic backgrounds about evolutionary equations with fermionic variables can be found in \cite{kuper85,wata89}, while rudiments of super algebras in \cite{rogers07}.

It is convenient to work with a super variable $\Phi=\Phi(x,\theta,t)$ 
where both $x$ and $t$ are usual commutative space-time coordinates, while $\theta$ is an anti-commutative 
independent variable such that $\theta^2=0$. The super variable $\Phi$ can be either fermionic (anti-commutative) or bosonic (commutative)
in this section.
The $k$-th order derivative of $\Phi$ with respect to $x$ is denoted by $\Phi_k \equiv (\partial_x^k\Phi)(k\in\bbbn)$\footnote{$\bbbn$ denotes the set of non-negative integers, $\{0,1,2,3,\cdots\}$.}, (We often omit the subscript of $\Phi_0$, and simply write $\Phi$.) and $\mathcal{D}$ stands for the super derivative defined by $\mathcal{D} \equiv \partial_\theta+\theta\partial_x$ with the property $\mathcal{D}^2=\partial_x$.

Let $\fc$ denote the super ($\mathbb{Z}_2$-graded) differential ring of smooth functions or simply polynomials in $\Phi$ and its (super) derivatives. As a matter of fact, every element in $\fc$ is supersymmetric, namely invariant under the transformation
\begin{equation}\label{susy}
\left\{\begin{aligned}
x & \mapsto x-\eta\theta \\
\theta & \mapsto\theta+\eta 
\end{aligned}\right.\qquad\mbox{($\eta$ is a fermionic parameter).}
\end{equation}
The multiplication on $\fc$ is super commutative, that is, for $P,Q\in\mathfrak{C}$
\begin{equation*}
PQ=(-1)^{|P||Q|}QP 
\end{equation*}
where $|~|$ denotes the parity of a super object\footnote{Formulas involving parities are extended to all elements in $\mathfrak{C}$ by linearity \cite{rogers07}.}. Derivatives on $\fc$, such as the super derivative $\cD$ and partial derivatives with respect to $\Phi_k$ or $(\cD\Phi_k)$, obey super Leibniz rules
\begin{eqnarray*}
\cD (PQ) = (\cD P)Q + (-1)^{|P|}P(\cD Q) , \quad
\frac{\partial }{\partial\dm}(PQ) = \frac{\partial P}{\partial\dm} Q + (-1)^{|\dm||P|}P\frac{\partial Q}{\partial\dm}
\end{eqnarray*}
where $\dm$ represents $\Phi_k$ or $(\cD\Phi_k)$.  For any element $P\in\fc$, there must be a least non-negative integer $k$ such that 
\begin{equation*}
\frac{\partial P}{\partial\Phi_l} = \frac{\partial P}{\partial(\cD \Phi_l)} = 0 \quad(l>k) .
\end{equation*}
If, in addition, we have 
\begin{equation*}
\frac{\partial P}{\partial(\cD \Phi_k)} = 0 ,
\end{equation*}
then $P$ is said to be {\em of order $k$}. Otherwise, we say $P$ is {\em of order $k{1\over 2}$}.

\begin{definition}
A partial differential operator of the form
\begin{equation}\label{svf}
\bm{P} = \sum_{k=0}\left( P_{2k}\frac{\partial}{\partial\Phi_k} + P_{2k+1}\frac{\partial}{\partial(\cD \Phi_k)}\right)  \quad (P_j\in \fc,~\forall j\in \bbbn)
\end{equation}
is called {\it a vector field on $\fc$}.
\end{definition}

Regarding the vector field $\bm{P}$ defined by \eqref{svf}, if there exists $|\bm{P}|\in \mathbb{Z}_2$ such that
\begin{equation*}
|\bm{P}| = (|P_j| - |\Phi| - j)\pmod 2 ,
\end{equation*}
where $P_j$ is an arbitrary non-zero coefficient, then $\bm{P}$ is {\em of parity $|\bm{P}|$}. All vector fields form a super vector space, denoted by $\mathfrak{V}$. In fact, $\mathfrak{V}$ is a super Lie algebra with the bracket defined by
\begin{equation*}
[\bm{P},\bm{Q}] = \bm{P}\bm{Q} - (-1)^{|\bm{P}||\bm{Q}|}\bm{Q}\bm{P} ,
\end{equation*}
and the super Jacobi identity reads as
\begin{equation*}
(-1)^{|\bm{P}||\bm{R}|}[\bm{P},[\bm{Q},\bm{R}]] + (-1)^{|\bm{Q}||\bm{P}|}[\bm{Q},[\bm{R},\bm{P}]] + (-1)^{|\bm{R}||\bm{Q}|}[\bm{R},[\bm{P},\bm{Q}]] = 0 \quad (\bm{P},\bm{Q},\bm{R}\in\mathfrak{V}) .
\end{equation*}

\begin{example} The prolongation of super derivative $\cD $ on $\mathfrak{C}$ is a vector field of parity $1$ given by
\begin{equation*}
\sum_{k=0} \left((\cD \Phi_k)\frac{\partial}{\partial \Phi_k} + \Phi_{k+1}\frac{\partial}{\partial (\cD \Phi_k)} \right) .
\end{equation*}
We shall still denote it by $\cD $ without confusion. 

Given a vector field $\bm{P}$ defined by \eqref{svf} of parity $|\bm{P}|$, by straightforward calculation we obtain 
\begin{equation*}
[\bm{P},\cD ] = \sum_{k=0}\left(\left(P_{2k+1} - (-1)^{|\bm{P}|}(\cD P_{2k})\right)\frac{\partial}{\partial \Phi_k} + \left(P_{2(k+1)} - (-1)^{|\bm{P}|}(\cD P_{2k+1})\right)\frac{\partial}{\partial (\cD \Phi_k)}\right) .
\end{equation*}
Hence, $\bm{P}$ commutes with $\cD $, i.e. $[\bm{P},\cD ] = 0$ if and only if 
\begin{equation*}
P_{j} = (-1)^{|\bm{P}|}(\cD P_{j-1})= (-1)^{j|\bm{P}|}(\cD ^jP_0) \quad  (j\in \mathbb{N}).
\end{equation*}
\end{example}

\begin{definition}\label{def:evo}
Given $P\in\mathfrak{C}$,  a vector field
\begin{equation*}
\sum_{k=0} \left( (\partial_x^kP)\frac{\partial}{\partial \Phi_k} + (-1)^{(|P|-|\Phi|)}(\cD \partial_x^kP)\frac{\partial}{\partial (\cD \Phi_k)} \right)
\end{equation*}
is called an {\it evolutionary} vector field, denoted by $\v_P$, whose parity is $(|P|-|\Phi|)\pmod 2$.
\end{definition}
We call $P\in\fc$ the {\it characteristic} of $\v_P$. When we talk about evolutionary vector fields, we often simply refer to their characteristics. The super Lie bracket of two evolutionary vector fields is still evolutionary, and indeed, for $P,Q\in\mathfrak{C}$ we have
\begin{equation}\label{veclie}
[\v_P,\v_Q] =\v_H, \quad H= \v_P(Q) - (-1)^{|\v_P||\v_Q|}\v_Q(P).
\end{equation}
This defines a super Lie bracket and thus the set of evolutionary vector fields
is a super Lie subalgebra of $\fv$. We are going to rewrite the Lie bracket in terms of Fr\'{e}chet derivatives.
\begin{definition}\label{def:fre}
For any $Q\in\fc$, we define a linear differential operator
\begin{equation}\label{frechet}
\mathrm{D}_Q = \sum_{k=0} \left( (-1)^{|\Phi|(|Q|-|\Phi|)}\frac{\partial Q}{\partial \Phi_k}\partial_x^k + (-1)^{(|\Phi|+1)(|Q|-|\Phi|+1)}\frac{\partial Q}{\partial (\cD \Phi_k)}\cD \partial_x^k \right) ,
\end{equation}
and call it the {\it Fr\'{e}chet derivative} of $Q$. 
\end{definition}
Note that the parity of the Fr\'{e}chet derivative for $Q$ is $(|Q|-|\Phi|)\pmod 2$, which is the same as 
$|\v_Q|$. In fact, we have
\begin{equation*}
\v_P(Q) = (-1)^{|\v_P||\v_Q|}\sum_{k=0} \left( (-1)^{|\Phi|(|Q|-|\Phi|)}\frac{\partial Q}{\partial \Phi_k}(\partial_x^kP) + (-1)^{(|\Phi|+1)(|Q|-|\Phi|+1)}\frac{\partial Q}{\partial (\cD \Phi_k)}(\cD \partial_x^kP) \right) .
\end{equation*}
Comparing to equation \eqref{frechet}, it becomes to
\begin{equation*}
\v_P(Q) = (-1)^{|\v_P||\v_Q|}\mathrm{D}_Q(P) .
\end{equation*}
Therefore, we endow $\fc$ a super Lie bracket defined by
\begin{equation}\label{def:com}
[P,Q] = \mathrm{D}_P(Q) - (-1)^{|\mathrm{D}_P||\mathrm{D}_Q|}\mathrm{D}_Q(P) ,
\end{equation}
whose parity is $|[P,Q]| = (|P|+|Q|-|\Phi| )\pmod 2$. This super Lie algebra 
is isomorphic to the super Lie subalgebra of  evolutionary vector fields since we have
\begin{equation*}
[\v_P,\v_Q] = - \v_{[P,Q]} .
\end{equation*}
The super Jacobi identity on $\mathfrak{C}$ is given by
\begin{equation}\label{pjacobi}
(-1)^{|\mathrm{D}_P||\mathrm{D}_R|}[P,[Q,R]] + (-1)^{|\mathrm{D}_Q||\mathrm{D}_P|}[Q,[R,P]] + (-1)^{|\mathrm{D}_R||\mathrm{D}_Q|}[R,[P,Q]] = 0 \quad (P,Q,R\in\mathfrak{C}) .
\end{equation}

\begin{definition}
For a $N=1$ supersymmetric evolutionary equation\footnote{The temporal variable $t$ is bosonic, so $K$ must have the same parity as $\Phi$.}
\begin{equation}\label{susyeq}
\Phi_t = K, \quad (K\in\mathfrak{C}~\mbox{and}~|K|=|\Phi|),
\end{equation}
any $Q\in\mathfrak{C}$ such that $[K,Q]=0$ is called an {\it infinitesimal symmetry}\footnote{There is no constraints on parities of symmetries.}(abbreviated to symmetries) of equation \eqref{susyeq}.
\end{definition}

\begin{example}
Equation $\eqref{susyeq}$ has at least three symmetries, namely 
\begin{align*}
\Phi_1,\quad K\quad \mbox{and} \quad S\equiv (\cD \Phi)-2\theta\Phi_1 ,
\end{align*}
where $K$ is obviously certified by a trivial identity $[K,K]=0$, and the other two are proved as follows:

The Fr\'{e}chet derivatives of $K$, $\Phi_1$ and $S$ are respectively given by
\begin{equation*}
\mathrm{D}_K = \sum_{k=0}  \left(\frac{\partial K}{\partial \Phi_k} \partial_x^k + (-1)^{(|\Phi|+1)} \frac{\partial K}{\partial (\cD \Phi_k)} \cD \partial_x^k \right),~~\mathrm{D}_{\Phi_1} = \partial_x~~\mbox{and}~~\mathrm{D}_S = \cD  - 2\theta\partial_x .
\end{equation*}
Notice that $|\mathrm{D}_K|=0$, and according to the Lie bracket \eqref{def:com}, we have
$[K,Q] = \mathrm{D}_K(Q) - \mathrm{D}_Q(K)$. Hence, we easily conclude that $\Phi_1$ is
a symmetry of equation \eqref{susyeq} since
\begin{align*}
\mathrm{D}_K(\Phi_1) = \partial_x(K) = \mathrm{D}_{\Phi_1}(K) ,
\end{align*}
that is, $[K, \Phi_1]=0$. For $S$, we compute it in the same way:
\begin{align*}
\mathrm{D}_K(S) =& \sum_{k=0}  \left(\frac{\partial K}{\partial \Phi_k} ((\cD \Phi_k) - 2\theta\Phi_{k+1}) + (-1)^{(|\Phi|+1)} \frac{\partial K}{\partial (\cD \Phi_k)} (2\theta(\cD \Phi_{k+1}) - \Phi_{k+1}) \right)\\
=& (\cD K) - 2\theta(\partial_xK) = \mathrm{D}_{S}(K). 
\end{align*}
In fact, $\Phi_1$ indicates the invariance of equation \eqref{susyeq} under $x$-translation, while $K$ represents the invariance under $t$-translation, and $S$ is the infinitesimal form of the transformation \eqref{susy}.
\end{example}

As an immediate corollary of the super Jacobi identity \eqref{pjacobi}, {\it if $Q_1$ and $Q_2$ are symmetries of equation \eqref{susyeq}, then so is $[Q_1,Q_2]$.} Therefore, all symmetries of equations \eqref{susyeq} constitute a super Lie subalgebra of $\mathfrak{C}$, which is the centralizer of $K$ in $\mathfrak{C}$ and denoted by $C_{\mathfrak{C}}(K)$.

Finally, we end the section with the definition of integrability used in this paper.
\begin{definition}
The supersymmetric equation \eqref{susyeq} is said to be {\it integrable} if $C_{\mathfrak{C}}(K)$ is infinite-dimensional. 
\end{definition}

\section{A complete list of scalar $\lambda$-homogeneous ($\lambda>0$) $N=1$ supersymmetric integrable evolutionary equations}\label{sec:3}
In this section, we give a global classification result for scalar
$\lambda$-homogeneous
$N=1$ supersymmetric evolutionary equations like equation \eqref{susymodel} when $\lambda>0$. 
It should be reminded {\it the super variable $\Phi$ is assumed to be fermionic} as usual.
The completeness 
is in the sense that every other such integrable equation is contained in the 
symmetry hierarchy of one of these equations presented here.
A $n$th order supersymmetric evolutionary equation \eqref{susymodel}
is said to be {\it $\lambda$-homogeneous} if it is invariant under the scaling 
transformation 
\begin{equation*}
(x,\theta,t,\Phi) \mapsto (e^{-\epsilon}x,e^{-\epsilon/2}\theta,e^{-n\epsilon}t,e^{\lambda\epsilon}\Phi),\quad (\epsilon\in\mathbb{R} ) .
\end{equation*}
For example, the supersymmetric KdV equation \cite{manin85,mathieu88}
\begin{equation*}
\Phi_t = \Phi_3 + 3\Phi_1(\cD \Phi) + 3\Phi(\cD \Phi_1)
\end{equation*}
is a $3/2$-homogeneous integrable equation.

By letting $\Phi = \xi(x,t)+\theta v(x,t)$, where $\xi$ and $v$ are often
referred to as components of the super variable $\Phi$, any scalar $N=1$ supersymmetric 
equation in $\Phi$ can be rewritten as a system in $\xi$ and $v$. When 
$\xi$ vanishes, the reduced equation in $v$ is called the bosonic limit of that 
supersymmetric model. A supersymmetric equation is said to be a {\it 
nontrivial} extension of its bosonic limit if the corresponding system in $\xi$ 
and $v$ are 
truly coupled, otherwise it is categorized as a {\it trivial} extension since it 
is essentially equivalent to a certain classical equation under simple changes 
of variables, and trivially inherits integrability from the latter. However, it 
is noticed that trivial extensions happen to be involved with supersymmetric 
extensions of matrix models, or conformal field theories coupled to gravity 
\cite{becker93}, so they are worthy of more investigations as well. 
In the rest of this section, we list the nontrivial cases as well as their 
component forms, for which we use the notations $v_k\equiv (\partial_x^kv)$ and 
$\xi_k \equiv(\partial_x^k\xi)$ $(k\in\mathbb{N})$. For the trivial cases,
we list them in appendix \ref{app:a}.

\vspace{0.5cm}
There are no second order nontrivial $\lambda$-homogeneous $N=1$ supersymmetric integrable equations.

\subsection{Third order nontrivial integrable equations}
\begin{itemize}
\item[(i).] Supersymmetric KdV equation \cite{manin85,mathieu88} 
($\lambda=\frac{3}{2}$)
\begin{equation}\label{skdv}
\Phi_t = \Phi_{3} + 3\Phi_1(\cD \Phi) + 3\Phi(\cD \Phi_1) .
\end{equation}
We rewrite it in components:
\begin{equation*}
\left\{\begin{aligned}
v_t = v_{3} &+ 6v_1v + 3\xi_{2}\xi \\
\xi_t = \xi_{3} &+ 3\xi_1 v + 3\xi v_1 .
\end{aligned}\right.
\end{equation*}
As a reduction of the supersymmetric KP hierarchy introduced by Manin and Radul \cite{manin85}, the supersymmetric KdV equation \eqref{skdv} can be encoded in the Lax equation $L_t = [4(L^{3/2})_{+},L]$ where $L = \partial_x^2+\Phi\cD $. Through $r$-matrix approach, its bi-Hamiltonian structure  was established by Oevel and Popowicz \cite{oevel91}, and its recursion operator is consequently given by
\begin{equation}\label{kdvrec}
\Re=\Big(\partial_x^2\cD  + 2\Phi\partial_x + 2\partial_x\Phi + 
\cD \Phi\cD \Big)\partial_x^{-1} \Big(\partial_x\cD  + 
\Phi\Big)\partial_x^{-1} .
\end{equation}

\item[(ii).] Supersymmetric potential KdV equation ($\lambda=\frac{1}{2}$)
\begin{equation}\label{spkdv}
\Phi_t = \Phi_{3} + 3\Phi_1(\cD \Phi_1)
\end{equation}
and it is rewritten in components as
\begin{equation*}
\left\{\begin{aligned}
v_t = v_{3} &+ 3v_1^2 + 3\xi_{2}\xi_1 \\
\xi_t = \xi_{3} &+ 3\xi_1 v_1.
\end{aligned}\right.
\end{equation*}
If equation \eqref{spkdv} holds, then $\Phi_1$ solves the supersymmetric KdV 
equation \eqref{skdv}. The recursion operator of equation \eqref{spkdv} is 
\begin{equation*}
\Re=\partial_x^{-1}\Big(\partial_x^2\cD  + 2\Phi_1\partial_x + 
2\partial_x\Phi_1 + \cD \Phi_1\cD \Big)\partial_x^{-1} 
\Big(\partial_x\cD  + \Phi_1\Big) .
\end{equation*}

\item[(iii).] Supersymmetric modified KdV equation \cite{mathieu88} 
($\lambda=\frac{1}{2}$)
\begin{equation}\label{smkdv}
\Phi_t = \Phi_{3} - 3\Phi_1(\cD \Phi)^2 - 3\Phi(\cD \Phi_1)(\cD \Phi)
\end{equation}
and its component form is
\begin{equation*}
\left\{\begin{aligned}
v_t = v_{3} &- 6v_1v^2 - 3\xi_{2}\xi v - 3\xi_1\xi v_1 \\
\xi_t = \xi_{3} &- 3\xi_1v^2 - 3\xi v_1 v.
\end{aligned}\right.
\end{equation*}

If $\Phi$ is a solution to equation \eqref{smkdv}, then $\Phi_1 - 
\Phi(\cD \Phi)$, which is referred to as the super Miura transformation 
\cite{mathieu88}, solves the supersymmetric KdV equation \eqref{skdv}.
Substituting the super Miura transformation into the recursion operator 
\eqref{kdvrec} and factorizing it produces the recursion operator  
\cite{liu10} of equation \eqref{smkdv}
\begin{equation*}
\Re=4\partial_x \Phi\cD ^{-1}\Phi\cD  - 
(\partial_x - \cD \Phi\cD ^{-1}\Phi - 
2\partial_x \Phi\partial_x^{-1}\Phi)\cD (\cD -\Phi\cD ^{-1}
\Phi - 2\Phi\partial_x^{-1}\Phi\cD ) .
\end{equation*}

\item[(iv).] Supersymmetric third order Burgers equation \cite{tian10} 
($\lambda=\frac{1}{2}$)
\begin{equation}\label{s3burgers}
\Phi_t = \Phi_{3} + 3\Phi_1(\cD \Phi_1) + 3\Phi(\cD \Phi_{2}) + 3\Phi(\cD \Phi_1)(\cD \Phi)
\end{equation}
and in components it can be written as
\begin{equation*}
\left\{\begin{aligned}
v_t = v_{3} &+ 3v_{2}v + 3v_1^2 + 3v_1v^2 + 3\xi_{3}\xi + 3\xi_{2}\xi_1 + 3\xi_{2}\xi v+ 3\xi_1\xi v_1 \\
\xi_t = \xi_{3} &+ 3\xi_1 v_1 + 3 \xi v_{2} + 3\xi v_1 v .
\end{aligned}\right.
\end{equation*}

Equation \eqref{s3burgers} does not commute with the trivial supersymmetric Burgers equation \eqref{bsburgers} in Appendix A, in other words, its right hand side is not a symmetry of equation \eqref{bsburgers}. In fact, the third order symmetry of equation \eqref{bsburgers} should be
\begin{equation*}
\Phi_3+3\Phi_2(\cD \Phi)+3\Phi_1(\cD \Phi_1)+3\Phi_1(\cD \Phi)^2 .
\end{equation*}
Until now, we have neither a recursion operator nor a master symmetry for 
equation \eqref{s3burgers}.
\end{itemize}
\subsection{Fifth order nontrivial integrable equations}
\begin{itemize}
\item[(i).] Supersymmetric Sawada-Kotera equation \cite{tian09} 
($\lambda=\frac{3}{2}$)
\begin{equation}\label{ssk}
\Phi_t = \Phi_{5} + 5\Phi_{3}(\cD \Phi) + 5\Phi_{2}(\cD \Phi_1) + 5\Phi_1(\cD \Phi)^2
\end{equation}
and  in components it is of the form
\begin{equation*}
\left\{\begin{aligned}
v_t = v_{5} &+ 5v_{3}v + 5v_{2}v_1 + 5v_1v^2 - 5\xi_{3}\xi_1 \\
\xi_t = \xi_{5} &+ 5\xi_{3}v + 5\xi_{2}v_1 + 5\xi_1 v^2 .
\end{aligned}\right.
\end{equation*}

The recursion operator of the supersymmetric Sawada-Kotera equation \eqref{ssk} was deduced from its Lax equation \cite{tian09}, and later factorized by Popowicz \cite{popo09} into two ``odd'' Hamiltonian operators, which allow us to rewrite the recursion operator in a neater form as
\begin{equation*}
\Re=\Big(\partial_x^2\cD  + 2\Phi\partial_x + 2\partial_x\Phi + 
\cD \Phi\cD \Big) \partial_x^{-1} \Big(\partial_x^2\cD  + 
2\Phi\partial_x + 2\partial_x\Phi + \cD \Phi\cD \Big) 
\cD ^{-1} (\partial_x\cD  + \Phi)^2\cD ^{-1} .
\end{equation*}

\item[(ii).] Supersymmetric fifth order KdV equation \cite{tian10} 
($\lambda=\frac{3}{2}$)
\begin{equation}\label{s5kdv}
\Phi_t = \Phi_{5} + 10\Phi_{3}(\cD \Phi) + 15\Phi_{2}(\cD \Phi_1) + 5\Phi_1(\cD \Phi_{2}) + 15\Phi_1(\cD \Phi)^2 + 15\Phi(\cD \Phi_1)(\cD \Phi)
\end{equation}
and in components it can be written as
\begin{equation*}
\left\{\begin{aligned}
v_t = v_{5} &+ 10v_{3}v + 20 v_{2}v_1 + 30 v_1v^2 - 5\xi_{3}\xi_1 + 15\xi_{2}\xi v + 15\xi_1\xi v_1 \\
\xi_t = \xi_{5} &+ 10\xi_{3}v + 15\xi_{2}v_1 + 5\xi_1v_{2} + 15\xi_1v^2 + 15\xi v_1v .
\end{aligned}\right.
\end{equation*}

Applying the recursion operator \eqref{kdvrec} to the supersymmetric KdV equation \eqref{skdv}, its fifth order symmetry is obtained
\begin{equation}\label{kdv5sym}
\Phi_5 + 5\Phi_3(\cD \Phi) + 10\Phi_2(\cD \Phi_1) + 10\Phi_1(\cD \Phi_2) + 5\Phi(\cD \Phi_3) + 10\Phi_1(\cD \Phi)^2 + 20\Phi(\cD \Phi_1)(\cD \Phi) .
\end{equation}
This fact implies equation \eqref{s5kdv} does not belong to the Manin-Radul's supersymmetric KdV hierarchy leading by equation \eqref{skdv}. In fact, equation \eqref{s5kdv} has a sixth order recursion operator \cite{tian10}, rewritten as 
\begin{equation*}
\Re = \mathcal{A}\partial_x^{-1}\Big(\partial_x^3\cD + 6\Phi\partial_x^2 + 4(\cD\Phi)\partial_x\cD + 8\Phi_1\partial_x + 2(\cD\Phi_1)\cD + 3\Phi_2 + 9\Phi(\cD\Phi)\Big)\partial_x^{-1} \mathcal{A}\cD^{-1} ,
\end{equation*}
where $\mathcal{A} \equiv \partial_x^2\cD + 3\Phi\partial_x + (\cD\Phi)\cD + 2\Phi_1$.

\item[(iii).] Supersymmetric Fordy-Gibbons equation \cite{liu10} 
($\lambda=\frac{1}{2}$)
\begin{align}
\Phi_t = \Phi_{5} & - 5\Phi_{2}(\cD \Phi_{2}) - 5\Phi_1(\cD \Phi_{3}) - 5\Phi_{3}(\cD \Phi)^2 - 10\Phi_{3}\Phi_1\Phi - 5\Phi_{2}(\cD \Phi_1)(\cD \Phi) \nonumber\\
& - 10\Phi(\cD \Phi_{2})(\cD \Phi_1) - 5\Phi_1(\cD \Phi_{2})(\cD \Phi) - 5\Phi_1(\cD \Phi_1)^2 - 10\Phi_{2}\Phi_1\Phi(\cD \Phi) \nonumber \\
& + 10\Phi_1(\cD \Phi_1)(\cD \Phi)^2 - 10\Phi(\cD \Phi_1)^2(\cD \Phi) + 5\Phi_1(\cD \Phi)^4 \label{sfg}
\end{align}
and in components it is of the form
\begin{equation*}
\left\{\begin{aligned}
v_t = v_{5} &- 5v_{3}v_1 - 5v_{3}v^2 - 5v_{2}^2 - 20v_{2}v_1v - 5v_1^3 + 5v_1v^4 - 5\xi_{4}\xi_1 - 5\xi_{3}\xi_{2} \\
& - 5\xi_{3}\xi_1v - 5\xi_{2}\xi_1 v_1 - 10\xi_{2}\xi(v_{2} + v_1v) -10\xi_1\xi(v_{3} + v_{2}v + v_1^2)  \\
\xi_t = \xi_{5} &- 5\xi_{3}v^2 - 10\xi_{3}\xi_1\xi - 5\xi_{2}(v_{2}+v_1 v) - 10\xi_{2}\xi_1\xi v - 5\xi_1(v_{3}+v_{2}v) \\
& - 5\xi_1(v_1^2-2v_1v^2-v^4) - 10\xi(v_{2}v_1+v_1^2v) .
\end{aligned}\right.
\end{equation*}

{ If $\Phi$ is a solution to equation \eqref{sfg}, then $\Phi_1 - 
\Phi(\cD \Phi)$ solves the supersymmetric Sawada-Kotera equation 
\eqref{ssk}.} With the super Miura transformation, the bi-Hamiltonian structure 
of equation \eqref{sfg} was established \cite{liu10} from that of the 
supersymmetric Sawada-Kotera equation \eqref{ssk}. The recursion operator of 
equation \eqref{sfg} is given by 
\begin{equation*}
\Re=\cD \mathcal{B}\partial_x^{-1}\mathcal{B}^\dagger\cD \mathcal{B}
\mathcal { 
D}^{-1}(\cD -\Phi)\cD (\cD +\Phi)(\cD -\Phi)\mathcal 
{D}(\cD +\Phi) \cD ^{-1}\mathcal{B}^\dagger ,
\end{equation*}
where $\mathcal{B} \equiv \partial_x - \Phi\cD + 2(\cD \Phi)$ and $\mathcal{B}^\dagger$ is the formal adjoint of $\mathcal{B}$, i.e. $\mathcal{B}^\dagger = - \partial_x + \Phi\cD  + (\cD \Phi)$.

\item[(iv).] Supersymmetric fifth order modified KdV equation 
($\lambda=\frac{1}{2}$)
\begin{align}
\Phi_t = \Phi_{5} &+ 5\Phi_{3}(\cD \Phi_1) - 5\Phi_1(\cD \Phi_{3}) - 10\Phi_{3}(\cD \Phi)^2 - 15\Phi_{3}\Phi_1\Phi - 15\Phi_{2}(\cD \Phi_1)(\cD \Phi) \nonumber \\
& - 10\Phi_1(\cD \Phi_{2})(\cD \Phi) - 10\Phi_1(\cD \Phi_1)^2 - 15\Phi(\cD \Phi_{2})(\cD \Phi_1) - 15\Phi_{2}\Phi_1\Phi(\cD \Phi) \nonumber \\
& + 15\Phi_1(\cD \Phi_1)(\cD \Phi)^2 - 15\Phi(\cD \Phi_1)^2(\cD \Phi) + 15\Phi_1(\cD \Phi)^4 + 15\Phi(\cD \Phi_1)(\cD \Phi)^3 \label{s5mkdv}
\end{align}
and in components it is of the form
\begin{equation*}
\left\{\begin{aligned}
v_t = v_{5} & - 10v_{3}v^2 - 40v_{2}v_1v - 10v_1^3 + 30v_1v^4 - 5\xi_{4}\xi_1 - 5\xi_{3}\xi_{2} - 5\xi_{3}\xi_1v \\
& - 5\xi_{2}\xi_1v_{2} - 15\xi_{2}\xi(v_{2}+v_1v-v^3) - 15\xi_1\xi(v_{3}+v_{2}v+v_1^2-3v_1v^2) \\
\xi_t = \xi_{5} & + 5\xi_{3}(v_1-2v^2) - 15\xi_{3}\xi_1\xi - 15\xi_{2}v_1v - 15\xi_{2}\xi_1\xi v - 5\xi_1(v_{3}+2v_{2}v) \\
& - 5\xi_1(2v_1^2-3v_1v^2-3v^4) - 15\xi(v_{2}v_1+v_1^2v-v_1v^3) .
\end{aligned}\right.
\end{equation*}

Equation \eqref{s5mkdv} is not a symmetry for the supersymmetric 
modified KdV equation \eqref{smkdv}, whose fifth order symmetry is  given 
by 
\begin{align*}
\Phi_5 &  - 5\Phi_3(\cD \Phi)^2 - 15\Phi_2(\cD \Phi_1)(\cD \Phi) - 15\Phi_1(\cD \Phi_2)(\cD \Phi) - 10\Phi_1(\cD \Phi_1)^2 \\
& - 5\Phi(\cD \Phi_3)(\cD \Phi) - 10\Phi(\cD \Phi_2)(\cD \Phi_1) + 10\Phi_1(\cD \Phi)^4 + 20\Phi(\cD \Phi_1)(\cD \Phi)^3 .
\end{align*}
Equation \eqref{s5mkdv} is related to equation \eqref{s5kdv} by a Muira 
transformation: {\it supposing $\Phi$ is a solution to equation \eqref{s5kdv}, 
then $\Phi_1-\Phi(\cD \Phi)$ solves equation \eqref{s5kdv}.} Thus its
integrability immediately follows from that of equation \eqref{s5kdv} due to 
this relation. 
\end{itemize}

The main result of this paper is to show that the above list is complete for 
positive $\lambda$.
\begin{The}\label{the}
A scalar $\lambda$-homogeneous ($\lambda>0$) $N=1$ supersymmetric
equation  (\ref{susymodel}) is nontrivial integrable if and only if $\lambda= 
\frac{1}{2}$ or $\frac{3}{2}$ and the equation belongs
to one the preceding eight symmetry hierarchies.
\end{The}

In particular, unlike the classical scalar case \cite{sand98}, some equations such as the 
second order Burgers equation, the Ibragimov-Shabat equation and the Kaup-Kupershmidt equation 
have no nontrivial supersymmetric integrable analogues. On the other hand, there are new supersymmetric 
third order Burgers, fifth order KdV and modified KdV equations. Super symmetrisation 
both decreases and increases the number of integrable equations in nontrivial 
cases.

\section{$N=1$ Supersymmetric differential polynomials and symbolic representation}\label{sec:4}
In this section, we extend symbolic representation to the ring of 
supersymmetric differential polynomials, that is, the elements in $\fc$ are 
polynomials in $\Phi$ and its (super) derivatives. As in the classical case \cite{sand98}, we concentrate ourselves on supersymmetric 
differential monomials to develop symbolic representation, and all results are 
straightforwardly generalized to supersymmetric differential polynomials by linearity. 

A supersymmetric  differential monomial of the form 
\begin{equation*}
\Phi_{k_1}\Phi_{k_2}\cdots \Phi_{k_m}(\cD \Phi_{l_1})(\cD \Phi_{l_2})\cdots(\cD \Phi_{l_n})
\end{equation*}
is said to be {\it of degree $(m,n)$} explicitly indicating its dependences on 
variables of different parities: $m$ fermionic variables and $n$ bosonic 
variables, and sometimes we say that its {\it total degree} is $r=m+n>0$ to emphasize the whole multiplicity in $\Phi$
and its (super) derivatives. 

Let $\fr^{(m,n)}$ be the space linearly spanned by all supersymmetric differential 
monomials of degree $(m,n)$ with coefficients in $\bbbc$, then for a fixed positive integer $r$, 
\begin{equation*}
\fr^{r} = \bigoplus_{m+n=r}\fr^{(m,n)}\quad (m,n\in\bbbn)
\end{equation*}
is the space containing all monomials of total degree $r$ and their linear combinations, and hence
\begin{equation*}
\fr= \bigoplus_{r>0}\mathfrak{R}^{r}= \bigoplus_{m,n}\mathfrak{R}^{(m,n)} \quad (m,n\in\bbbn,\ m+n>0)
\end{equation*}
is the super commutative algebra of all supersymmetric differential polynomials in $\Phi$ and its (super) derivatives. It is noticed that $\bbbc\not\subset\fr$.

With the super Lie bracket \eqref{def:com}, $\fr$ is a graded super Lie algebra since
\begin{equation*}
 [\fr^r, \ \fr^{r^\prime}]\subset \fr^{r+r^\prime-1} \quad (r, r^\prime>0) .
\end{equation*}
Moreover, we have
\begin{equation}\label{grade}
 [\fr^{(m,n)},\ \fr^{(p,q)}]\subset \fr^{(m+p-1,n+q)}\bigoplus\fr^{(m+p+1,n+q-2)}, 
\end{equation}
where $n$ or $q$ are positive integers. In particular, $[\fr^{(1,0)},\ \fr^{(p,q)}]\subset \fr^{(p,q)}$.

To describe supersymmetric differential polynomials in symbolic language, the basic idea is to convert the fermionic variable $\Phi_k$ to $\phi\zeta^k$, and the bosonic one $(\cD \Phi_k)$ to $u z^k$, where $\phi$ or $u$ are used to count degrees, and $\Phi$'s subscript $k$ indicating the order of derivatives is replaced by powers of symbols. As results of this replacement, the differentiation with respect to $x$ is transformed to multiplications of symbols.  Variables in the same monomials of higher (total) degrees are assigned with different symbols, for instance, the quadratic term $\Phi_j\Phi_k$ is converted to $\phi^2\zeta_1^j\zeta_2^k$, while $(\cD \Phi_j)(\cD \Phi_k)$ is replaced by $u^2z_1^jz_2^k$. With these two simple examples, it is noticed that there are some confusions aroused by the super commutativity of multiplication. Specifically, we have
\begin{equation*}
\Phi_j\Phi_k = -\Phi_k\Phi_j\quad\mbox{and}\quad (\cD \Phi_j)(\cD \Phi_k) = (\cD \Phi_k)(\cD \Phi_j) .
\end{equation*}
Following above discussion, $-\Phi_k\Phi_j$ is replaced by $-\phi^2\zeta_1^k\zeta_2^j$, while $(\cD \Phi_k)(\cD \Phi_j)$ by $u^2z_1^kz_2^j$. From equivalent products, different symbolic expressions are obtained. To overcome these obstacles, the unique symbolic representation is defined to be  the average of all possible symbolic expressions. Hence, symbolic forms of $\Phi_j\Phi_k$ and $(\cD \Phi_j)(\cD \Phi_k)$ are respectively given by
\begin{equation*}
\frac{1}{2}\phi^2(\zeta_1^j\zeta_2^k - \zeta_2^j\zeta_1^k) \quad\mbox{and}\quad \frac{1}{2}u^2(z_1^jz_2^k+z_2^jz_1^k) .
\end{equation*}

\begin{definition}\label{defsymb}
The symbolic representation of a monomial in $\mathfrak{R}^{(m,n)}$ is defined by
\begin{equation*}
\Phi_{k_1}\Phi_{k_2}\cdots \Phi_{k_m}(\cD \Phi_{l_1})(\cD \Phi_{l_2})\cdots(\cD \Phi_{l_n}) \mapsto
\phi^m u^n\langle\zeta_1^{k_1}\zeta_2^{k_2}\cdots \zeta_m^{k_m}\rangle_{\mathcal{S}^{\zeta}_m} \langle z_1^{l_1}z_2^{l_2}\cdots z_n^{l_n}\rangle_{\mathcal{S}^z_n}
\end{equation*}
where $\langle~\rangle_{\mathcal{S}^{\zeta}_m}$ stands for anti-symmetrization of the product $\zeta_1^{k_1}\zeta_2^{k_2}\cdots \zeta_m^{k_m}$, i.e.
\begin{equation*}
\langle\zeta_1^{k_1}\zeta_2^{k_2}\cdots \zeta_m^{k_m}\rangle_{\mathcal{S}^{\zeta}_m} = \frac{1}{m!}\sum_{\sigma\in\mathcal{S}_m} \mathrm{sgn}(\sigma) \zeta_{\sigma(1)}^{k_1}\zeta_{\sigma(2)}^{k_2}\cdots \zeta_{\sigma(m)}^{k_m} ,
\end{equation*}
and $\langle ~\rangle_{\mathcal{S}^z_n}$ means to symmetrise the product $z_1^{l_1}z_2^{l_2}\cdots z_n^{l_n}$, namely
\begin{equation*}
\langle z_1^{l_1}z_2^{l_2}\cdots z_n^{l_n}\rangle_{\mathcal{S}^z_n} = \frac{1}{n!}\sum_{\sigma\in\mathcal{S}_n}z_{\sigma(1)}^{l_1}z_{\sigma(2)}^{l_2}\cdots z_{\sigma(n)}^{l_n} .
\end{equation*}
\end{definition}
A differential polynomial is a finite linear combination of monomials.
The symbolic representation of a polynomial $P\in\mathfrak{R}^{(m,n)}$ denoted 
by $\widehat{P}$ is in general of the form
\begin{equation*}
\widehat{P} = \phi^mu^n F(\zeta_1,\cdots,\zeta_m;z_1,\cdots,z_n), 
\end{equation*}
where the multi-variable polynomial $F$ is anti-symmetric with respect to its 
first $m$ arguments, but symmetric with respect to the others. Suppose a 
polynomial $Q\in\mathfrak{R}^{(p,q)}$ have the symbolic representation 
\begin{equation*}
\widehat{Q} = \phi^pu^q H(\zeta_1,\cdots,\zeta_p;z_1,\cdots,z_q) ,
\end{equation*}
then the product $PQ$  has the following symbolic representation:
\begin{equation}\label{prodsym}
\widehat{PQ} = \phi^{m+p}u^{n+q}\langle F(\zeta_1,\cdots,\zeta_m;z_1,\cdots,z_n)H(\zeta_{m+1},\cdots,\zeta_{m+p};z_{n+1},\cdots,z_{n+q})\rangle_{\mathcal{S}^{\zeta}_{m+p},\mathcal{S}^{z}_{n+q}} ，
\end{equation}
where $\langle~\rangle_{\mathcal{S}^{\zeta}_{m+p},\mathcal{S}^{z}_{n+q}} $ means 
to anti-symmetrize $\zeta_j$'s, and meanwhile symmetrize $z_i$'s. 

This defines the corresponding multiplication for the symbols. 
Formula \eqref{prodsym} is consistent with Definition \ref{defsymb}.  Symbolic 
representations of super differential monomials can be deduced from $\Phi_k 
\mapsto \phi \zeta_1^k$, $\cD  \Phi_k\mapsto u z_1^k$ and the 
multiplication rule \eqref{prodsym}.

All operations for the ring of supersymmetric differential polynomials can
be carried out with symbols in much simpler way.
Partial derivatives of $P$ with 
respect to $\Phi_k$ or $(\cD \Phi_k)$
can be calculated via its symbolic representation as follows:
\begin{eqnarray}
&&\widehat{\frac{\partial P}{\partial \Phi_k}} = 
m\phi^{m-1}u^n\frac{1}{k!}\frac{\partial^{k}F}{\partial 
\zeta_m^k}(\zeta_1,\cdots,\zeta_{m-1},0;z_1,\cdots,z_n);
\label{partial:1}\\
&&\widehat{\frac{\partial P}{\partial (\cD \Phi_k)}} = n\phi^m 
u^{n-1}\frac{1}{k!}\frac{\partial^{k}F}{\partial 
z_n^k}(\zeta_1,\cdots,\zeta_{m};z_1,\cdots,z_{n-1},0).\label{partial:2}
\end{eqnarray}
According to the definition of the super derivative $\cD$ and using the above two formulas, it follows that
\begin{align}
\widehat{(\cD P)} =& (-1)^{m-1}m\phi^{m-1}u^{n+1}\langle F(\zeta_1,\cdots, 
\zeta_{m-1},z_{n+1}; z_1,\cdots,z_n)\rangle_{\mathcal{S}^z_{n+1}}  \nonumber\\
& + (-1)^m n\phi^{m+1}u^{n-1} \langle F(\zeta_1,\cdots, \zeta_m; 
z_1,\cdots,z_{n-1}, \zeta_{m+1})\zeta_{m+1}\rangle_{\mathcal{S}^{\zeta}_{m+1}} . 
\label{superde}\\
= & (-1)^{m-1}\frac{m}{n+1}\phi^{m-1}u^{n+1}\sum_{i=1}^{n+1}F(\zeta_1,\cdots, 
\zeta_{m-1},z_{i};z_1,\cdots,\overline{z_i},\cdots, z_{n+1}) \nonumber\\
& + \frac{n}{m+1} \phi^{m+1}u^{n-1} \sum_{j=1}^{m+1} 
(-1)^{j-1}F(\zeta_1,\cdots,\overline{\zeta_j},\cdots, 
\zeta_{m+1};z_1,\cdots,z_{n-1}, \zeta_j)\zeta_j ,\label{superde2}
\end{align}
where $\overline{z_i}$ means that $z_i$ is omitted and $\overline{\zeta_j}$ is 
defined in the same manner. 
\begin{example}
Given $P = \Phi_3\Phi_2\Phi_1$, we calculate $\partial P/\partial \Phi_2$ using symbolic representation. The symbolic representation for $P$ is
\begin{equation*}
\widehat{P} = \phi^3F(\zeta_1,\zeta_2,\zeta_3) = \frac{1}{6}\phi^3(\zeta_1^3\zeta_2^2\zeta_3^1 + \zeta_2^3\zeta_3^2\zeta_1^1 + \zeta_3^3\zeta_1^2\zeta_2^1 - \zeta_1^3\zeta_3^2\zeta_2^1 - \zeta_2^3\zeta_1^2\zeta_3^1 - \zeta_3^3\zeta_2^2\zeta_1^1) .
\end{equation*} 
According to formula \eqref{partial:1}, we have
\begin{equation*}
\widehat{\frac{\partial P}{\partial \Phi_2}} = 3\phi^2\frac{1}{2!}\frac{\partial^2 F}{\partial \zeta_3^2} (\zeta_1,\zeta_2,0) = -  \frac{1}{2}\phi^2(\zeta_1^3\zeta_2^1-\zeta_2^3\zeta_1^1) ,
\end{equation*}
which conversely implies $\partial P/\partial \Phi_2 = -\Phi_3\Phi_1$, the expected result.
\end{example}

\begin{example}
For $P = \Phi_1\Phi$, let's calculate $(\cD P)$ and $(\partial_xP)$ through symbolic representation. Because
\begin{equation*}
\widehat{P} = \phi^2 F(\zeta_1,\zeta_2) = \frac{1}{2}\phi^2(\zeta_1^1\zeta_2^0 - \zeta_2^1\zeta_1^0) ,
\end{equation*}
it follows from formula \eqref{superde} that 
\begin{equation*}
\widehat{(\cD P)} = (-1)2\phi u F(\zeta_1,z_1) = -\phi u(\zeta_1^1z_1^0 - z_1^1\zeta_1^0),
\end{equation*}
which means $(\cD P) = -\Phi_1(\cD \Phi) + \Phi(\cD \Phi_1)$. Applying formula \eqref{superde} to $\widehat{(\cD P)}$ yields
\begin{equation*}
\widehat{(\partial_xP)} = -u^2 \langle z_2^1z_1^0\rangle_{\mathcal{S}^z_{2}} + u^2 \langle z_1^1z_2^0\rangle_{\mathcal{S}^z_{2}} + \phi^2 \langle \zeta_1^1\zeta_2^1\rangle_{\mathcal{S}^\zeta_{2}} - \phi^2 \langle \zeta_2^2\zeta_1^0\rangle_{\mathcal{S}^\zeta_{2}}
= - \frac{1}{2}\phi^2 (\zeta_2^2\zeta_1^0 - \zeta_1^2\zeta_2^0 ),
\end{equation*}
which give us $(\partial_xP) = \Phi_2\Phi$ as expected.
\end{example}

In the above example, we notice that $\widehat{(\partial_xP)}=(\zeta_1 +\zeta_2)\widehat{P}$ for $P= \Phi_1\Phi$. 
Here we give the general formula for it.

\begin{proposition}\label{prop:2}
Let $P\in\mathfrak{R}^{(m,n)}$ with symbolic representation $\widehat{P}= \phi^mu^n F(\zeta_1,\cdots,\zeta_m;z_1,\cdots,z_n)$. Then
\begin{equation*}\label{xder}
\widehat{(\partial_x P)} 
= \widehat{P}\left(\sum_{j=1}^m\zeta_j + \sum_{i=1}^nz_i\right).
\end{equation*}
Furthermore, 
\begin{equation*}
\widehat{(\partial_x^k P)} = \widehat{P}\left(\sum_{j=1}^m\zeta_j + \sum_{i=1}^nz_i\right)^k\quad (k\in\mathbb{N}) .
\end{equation*}
\end{proposition}
\begin{proof}
Since $(\partial_xP)=\cD (\cD P)$, the statement for $(\partial_x P)$ can be proved by applying the symbolic rule of super derivative \eqref{superde} or \eqref{superde2} to $\widehat{(\cD \Phi)}$. We rewrite $\widehat{\cD P}=\widehat{Q_1}+\widehat{Q_2}$,
where 
\begin{eqnarray*}
&&\widehat{Q_1} =  (-1)^{m-1}\frac{m}{n+1}\phi^{m-1}u^{n+1}\sum_{i=1}^{n+1}F(\zeta_1,\cdots, \zeta_{m-1},z_{i};z_1,\cdots,\overline{z_i},\cdots, z_{n+1}) ,\\
&&\widehat{Q_2} = \frac{n}{m+1} \phi^{m+1}u^{n-1} \sum_{j=1}^{m+1} (-1)^{j-1}F(\zeta_1,\cdots,\overline{\zeta_j},\cdots, \zeta_{m+1};z_1,\cdots,z_{n-1}, \zeta_j)\zeta_j 
\end{eqnarray*}
and $Q_1\in\mathfrak{R}^{(m-1,n+1)}$ while $Q_2\in\mathfrak{R}^{(m+1,n-1)}$.
We apply formula \eqref{superde} again to $\widehat{Q_1}$ and $\widehat{Q_2}$ separately, and carry out the symmetrization and 
anti-symmetrization accordingly. This leads to the proof of the statement. The calculation is straightforward
and we won't include all details.
\end{proof}

Note that formulas in Proposition \ref{prop:2} is essentially the same as the symbolic rule for derivatives with respect to $x$ in classical case \cite{sand98}.

We introduce the symbol for the super derivative $\cD$ as $\eta$, and consequently $\eta^2$ stands for $\partial_x$.
According to Definition \ref{def:fre} and using formulas \eqref{partial:1} and \eqref{partial:2}, we obtain the following statement 
for Fr\'{e}chet derivatives.

\begin{proposition}\label{prop:3}
Let $P\in\mathfrak{R}^{(m,n)}$ with symbolic representation $\widehat{P}= \phi^mu^n F(\zeta_1,\cdots,\zeta_m;z_1,\cdots,z_n)$.
The Fr\'{e}chet derivative of $P$ is represented in the following expression: 
\begin{equation}
\widehat{\mathrm{D}_P} = m\phi^{m-1} u^n F(\zeta_1,\cdots \zeta_{m-1}, \eta^2; z_1,\cdots,z_n) + n\phi^m u^{n-1} F(\zeta_1,\cdots \zeta_m; z_1,\cdots,z_{n-1},\eta^2)\eta ,
\end{equation}
where $\eta$ is the symbol for $\cD $.
\end{proposition}
\begin{example}
Let $P = \Phi_2\Phi(\cD \Phi_1)(\cD \Phi)$. It has the symbolic representation 
\begin{equation*}
\widehat{P} = \phi^2 u^2 F(\zeta_1,\zeta_2;z_1,z_2) = \frac{1}{4}\phi^2 u^2(\zeta_1^2\zeta_2^0 - \zeta_2^2\zeta_1^0)(z_1^1z_2^0 + z_2^1z_1^0) .
\end{equation*}
According to Proposition \ref{prop:3}, we have 
\begin{equation*}
\begin{aligned}
\widehat{\mathrm{D}_{P}}=& 2\phi u^2 F(\zeta_1,\eta^2;z_1,z_2) + 2\phi^2 u F(\zeta_1,\zeta_2;z_1,\eta^2)\eta \\
=& \frac{1}{2}\phi u^2 \zeta_1^2 (z_1^1z_2^0 + z_2^1z_1^0)\eta^0 - \frac{1}{2}\phi u^2 \zeta_1^0(z_1^1z_2^0 + z_2^1z_1^0) \eta^4 + \frac{1}{2}\phi^2 u(\zeta_1^2\zeta_2^0 - \zeta_2^2\zeta_1^0)z_1^1\eta\\
& + \frac{1}{2}\phi^2 u(\zeta_1^2\zeta_2^0 - \zeta_2^2\zeta_1^0)z_1^0 \eta^3 ,
\end{aligned}
\end{equation*}
which implies that the Fr\'{e}chet derivative of $P$ is given by
\begin{equation*}
\mathrm{D}_P=\Phi_2(\cD \Phi_1)(\cD \Phi) - \Phi(\cD \Phi_1)(\cD \Phi)\partial_x^2 + \Phi_2\Phi(D\Phi_1)\cD  + \Phi_2\Phi(D\Phi)\cD \partial_x .
\end{equation*}
\end{example}

Having derived the symbolic formulas for the super derivative $\cD$,  the usual derivative $\partial_x$ and 
Fr\'echet derivatives, we are in the place to write down the symbolic expression for $\mathrm{D}_P(Q)$
for $P\in \fr^{(m,n)}$ and $Q\in \fr^{(p,q)}$ as in classical case \cite{sand98}. 
However, the formula is rather long. We only give two special and simple cases, 
which we are going to use later.
\begin{proposition}\label{prop:4}
Let $P\in\mathfrak{R}^{(m,n)}$ and $P\mapsto\widehat{P}= \phi^mu^n F(\zeta_1,\cdots,\zeta_m;z_1,\cdots,z_n)$.
Then
\begin{equation}
\widehat{\mathrm{D}_P(\Phi_k)} = \widehat{P}\left(\sum_{j=1}^m\zeta_j^k +\sum_{i=1}^n z_i^k\right) .
\end{equation}
\end{proposition}
\begin{proof} Notice that $\widehat{\Phi_k}=\phi \zeta_1^k$, $\eta \widehat{\Phi_k}=u z_1^k$
and $\eta^2 \widehat{\Phi_k}=\phi \zeta_1^{k+1}$. Using Proposition \ref{prop:3}, we have
\begin{equation*}
\widehat{\mathrm{D}_P(\Phi_k)} = m\phi^{m} u^n \langle F(\zeta_1,\cdots \zeta_{m-1}, \zeta_m; z_1,\cdots,z_n)\zeta_m^k\rangle_{\mathcal{S}^{\zeta}_{m}} 
+ n\phi^m u^{n} \langle F(\zeta_1,\cdots \zeta_m; z_1,\cdots,z_{n-1},z_n)z_n^k\rangle_{\mathcal{S}^z_{n}} .
\end{equation*}
After symmetrisation and anti-symmetrisation, we obtain the required formula.
\end{proof}
Combining the result in this proposition and Proposition \ref{prop:2}, we have the following
formula for a super commutator  (cf. equation \eqref{def:com}), which plays an important 
role in later classification.
\begin{corollary}\label{cor:1}
Let $P\in\mathfrak{R}^{(m,n)}$ and $P\mapsto\widehat{P}= \phi^mu^n 
F(\zeta_1,\cdots,\zeta_m;z_1,\cdots,z_n)$.
The super commutator $[\Phi_k, P]$ is symbolically formulated 
as
\begin{align}\label{Gfun}
\widehat{[\Phi_k, P]} =\widehat{P} \left\{\left(\sum_{j=1}^m\zeta_j+\sum_{i=1}^n 
z_i\right)^k 
- \left(\sum_{j=1}^m\zeta_j^k + \sum_{i=1}^n z_i^k\right)\right\}.
\end{align}
\end{corollary}
\begin{proposition}\label{prop:40}
Let $P, Q\in \fr^{(1,1)}$ and $P\mapsto\widehat{P}=\phi u F(\zeta_1, z_1), Q\mapsto\widehat{Q}=\phi u H(\zeta_1, z_1) $. Then
\begin{align*}
 \widehat{\mathrm{D}_P(Q)}=& \phi u^2\langle F(\zeta_1+z_2,z_1) H(\zeta_1,z_2) 
+F(\zeta_1; z_1+z_2) H(z_1,z_2)\rangle_{\mathcal{S}^z_{2}}
\\
& -\phi^3\langle F(\zeta_1,\zeta_2+\zeta_3) H(\zeta_2,\zeta_3) \zeta_3
\rangle_{\mathcal{S}^{\zeta}_{3}}
\end{align*}
\end{proposition}
\begin{proof} The statement can be proved by directly applying
Proposition \ref{prop:3} and Proposition \ref{prop:2}.
\end{proof}

In the end of this section, we calculate a fifth order symmetry of the supersymmetric KdV equation \eqref{skdv} using the symbolic method to illustrate its mechanism. 
We know that the supersymmetric KdV equation is of the form
\begin{equation*}
\Phi_t = \Phi_3 + K_2 \quad\mbox{with}\quad K_2 = 3\Phi_1(\mathcal{D}\Phi) + 3\Phi(\mathcal{D}\Phi_1) \in\fr^{(1,1)} .
\end{equation*}
We compute its fifth order symmetry
\begin{equation*}
Q = \Phi_5 + Q_2 + Q_3 + \cdots \quad\mbox{with}\quad Q_i\in\fr^i\quad (i=2,3,\cdots) 
\end{equation*}
term by term according to its degree.
Firstly the quadratic term $Q_2\in \fr^2$ is determined by
\begin{equation*}
[\Phi_3,Q_2]+[K_2,\Phi_5] = 0 .
\end{equation*}
Following Corollary \ref{cor:1}, we convert it to symbolics
\begin{equation*}
\widehat{Q_2} \left((\zeta_1+z_1)^3 - \zeta_1^3 - z_1^3\right) = \widehat{K_2}\left((\zeta_1+z_1)^5 - \zeta_1^5 - z_1^5\right) ,
\end{equation*}
where $\widehat{K_2} = 3\phi u(\zeta_1^1z_1^0 + \zeta_1^0 z_1^1)$. Hence, we obtain
\begin{align*}
\widehat{Q_2} = \frac{(\zeta_1+z_1)^5 - \zeta_1^5 - z_1^5}{(\zeta_1+z_1)^3 - \zeta_1^3 - z_1^3}\widehat{K_2} = 5\phi u \left(\zeta_1^3z_1^0 + \zeta_1^0z_1^3\right) + 10 \phi u\left(\zeta_1^1z_1^2 + \zeta_1^2z_1^1\right) ,
\end{align*}
which implies $Q_2 = 5\Phi_3(\mathcal{D}\Phi) + 5\Phi(\mathcal{D}\Phi_3) + 10\Phi_1(\mathcal{D}\Phi_2) + 10\Phi_2(\mathcal{D}\Phi_1)$,
which is completely determined by the quadratic terms of the equation.

Secondly, the cubic terms $Q_3\in \fr^3$ is determined by
\begin{equation}\label{eq:q3}
[\Phi_3,Q_3] + [K_2,Q_2] = 0 .
\end{equation}
According to Proposition \ref{prop:40}, for $K_2, Q_2 \in \fr^{(1,1)}$ we have
\begin{align*}
\widehat{[Q_2,K_2]} =& 30\phi u^2 \left(\zeta_1^3 z_1^1 z_2^0 + \zeta_1^3 z_2^1 z_1^0 + 2\zeta_1^2 z_1^2 z_2^0 + 2\zeta_1^2 z_2^2 z_1^0 + 4\zeta_1^2 z_1^1 z_2^1 +\zeta_1^1 z_1^3 z_2^0 + \zeta_1^1 z_2^3 z_1^0 \right) \\
&+ 30\phi u^2 \left(4\zeta_1^1 z_1^2 z_2^1 + 4\zeta_1^1 z_2^2 z_1^1 + \zeta_1^0 z_1^3 z_2^1 + \zeta_1^0 z_2^3 z_1^1 + 2\zeta_1^1 z_1^2 z_2^2\right) .
\end{align*}
Thus $[Q_2,K_2]\in\mathfrak{R}^{(1,2)}$. According to equation \eqref{grade}, for any $P\in\mathfrak{R}^{(m,n)}$ we have $[\Phi_k,P]\in\mathfrak{R}^{(m,n)}$. So equation \eqref{eq:q3} implies $Q_3\in \mathfrak{R}^{(1,2)}$. From the symbolic representation of equation \eqref{eq:q3}, we deduce
\begin{equation*}
\widehat{Q_3} = \frac{\widehat{[Q_2,K_2]}}{(\zeta_1+z_1+z_2)^3-\zeta_1^3-z_1^3-z_2^3} =10\phi u^2\left(\zeta_1^1 z_1^0 z_2^0 + \zeta_1^0 z_1^1 z_2^0 + \zeta_1^0 z_2^1 z_1^0\right),
\end{equation*}
which gives us $Q_3 = 10\Phi_1(\mathcal{D}\Phi)^2 + 20\Phi(\mathcal{D}\Phi_1)(\mathcal{D}\Phi)$.

It is straightforward to check $[K_2,Q_3]=0$. Therefore, a fifth order symmetry of the supersymmetric KdV equation \eqref{skdv}
starting with $\Phi_5$ is the one given by
 \eqref{kdv5sym} generated by its recursion operator \eqref{kdvrec}.

\section{Classification of $\lambda$-homogeneous equations}\label{sec:5}
In this section, we use the symbolic representation to prove our main classification
Theorem \ref{the}.  The computations are remarkably similar to the classical commutative \cite{sand98}
and non-commutative \cite{olver00} cases. The key differences are
\begin{itemize}
 \item the polynomials arising in the symbolic
computations includes two sets of variables: both symmetrised and anti-symmetrised variables
under permutations;
\item while the bounds on the orders of the equation and its symmetries happen to be
the same as in the classical cases, the symbolic computation relies on 
whether or not the variables commute and this leads to rather different classification results
from the classical ones.
\end{itemize}

In \cite{sand98}, we gave extensive results about the mutual divisibility of certain
particular multivariate polynomials, called ``$G$-functions", which play a crucial
role in proving the (non-)existence of symmetries. In Corollary \ref{cor:1}, we
demonstrated that the same $G$-functions appear in the computation.
Thus all the results discussed in \cite[Section 5]{sand98} are immediately applicable.
\begin{definition}
The  \(G\)-functions are the {\rm (}commutative\/{\rm )} polynomials
\begin{equation*}
G_{m}^{(l)}(y_1, \cdots, y_{l}) = \bigl(y _1 + \cdots + 
y_{l}\bigr)^m-y_1^m  -\cdots - y_{l}^m \quad (l\geq 2).
\end{equation*}
\end{definition}
Formula \eqref{Gfun} can be rewritten in terms of $G$-function as
\begin{equation}\label{Gfk}
\widehat{[{\Phi_k}, Q]} = G_{k}^{(2l+n+1)}(\zeta_1, \cdots, \zeta_{2l+1}, z_1, 
\cdots, z_n)\  \widehat{Q}, \quad Q \in \fr^{(2l+1, n)}.
\end{equation}
An immediate consequence is the following known result:
\begin{proposition}\label{prop:5}
Consider the linear evolutionary equation $\Phi_t=\sum_{k=1}^n \lambda_k 
\Phi_k$, where $\lambda_k$'s are constants and $ \lambda_n\neq 0$. The space of its 
symmetries is
\begin{itemize}
\item \( \fr\) if and only if  \(n=1\);
\item \({\fr}^{1}\) if and only if \(n > 1\).
\end{itemize}
\end{proposition}
\begin{proof}
Let $Q=\sum_{l,p} Q_{(l,p)}$, where $Q_{(l,p)}\in 
{\fr}^{(l,p)}$. Because $\fr$ is a graded Lie algebra, and particularly 
$[\fr^{(1,0)},\ \fr^{(l,p)}]\subset\fr^{(l,p)}$, $Q$ is a symmetry of 
this equation if and only if  $[\sum_{k=1}^n\lambda_k \Phi_k,\ Q]=0$. 
From formula (\ref{Gfk}) it follows that
\[
\sum_{k=1}^n\lambda_k G_k^{l+p}(\zeta_1,\cdots,\zeta_l,z_1,\cdots,z_p)\widehat{Q_{(l,p)}} =0.
\]
Under the assumption, this holds if and only if either $n=1$ or $n> 1$ and $l+p=1$.
\end{proof}

We recall the divisibility properties of the \(G\)-functions proved in 
\cite{beuk97, sand98}.
\begin{The}\label{BeukId}
The symmetric polynomials $G_m^{(l)} (l\geq 2)$ can be factorized as
\[
G_m^{(l)}=t_m^{(l)} H_m^{(l)},
\]
where $(H_m^{(l)},H_n^{(l)})=1$ for all $n>m$, and \(t_m^{(l)}\) is one of the following polynomials.
\begin{itemize}
\item \(l = 2:\)
\begin{itemize}
\item[$\rhd$] \(m=0\pmod 2:\quad y_1y_2\)
\item[$\rhd$] \(m=3\pmod 6:\quad y_1y_2(y_1+y_2)\)
\item[$\rhd$] \(m=5\pmod 6:\quad y_1y_2(y_1+y_2)(y_1^2+y_1y_2+y_2^2) \)
\item[$\rhd$] \(m=1\pmod 6:\quad y_1y_2(y_1+y_2)(y_1^2+y_1y_2+y_2^2)^2\)
\end{itemize}
\item \(l = 3:\)
\begin{itemize}
\item[$\rhd$] \(m=0\pmod 2:\quad 1\)
\item[$\rhd$] \(m=1\pmod 2:\quad (y_1+y_2)(y_1 +y_3) (y_2 +y_3)\)
\end{itemize}
\item \(l \geq 4:\)  \(1\)
\end{itemize}
\end{The}

We now consider $\lambda$-homogeneous ($\lambda>0$) $N=1$ supersymmetric equations of the form
\begin{equation}
\Phi_t=K=K_1+K_2+K_3+\cdots,\label{hoeq}
\end{equation}
where $K_1 = \Phi_n~(n\geq 2)$ and 
\begin{equation*}
K_i \in \bigoplus_{l=0}^{\lfloor \frac{i-1}{2}\rfloor} \fr^{(2l+1,i-2l-1)}\subset \fr^i .
\end{equation*}
For each $K_i$, the degree of its symbolic representation $\widehat{K_i}$ is determine by 
\begin{equation*}
d_i = \left\{\begin{aligned}
& n-(i-1)\lambda-\frac{1}{2} && \mbox{($i$ is even.)}\\
& n-(i-1)\lambda && \mbox{($i$ is odd.)}
\end{aligned}\right.
\end{equation*}
Note that if $d_i$ is not in $\bbbn$ then $K_i=0$. This constraint restricts
equation \eqref{hoeq} to finite terms, and also reduces the number 
of relevant $\lambda$ to be finite.

Due to Proposition \ref{prop:5}, a nontrivial symmetry $Q\in\fr$ of equation \eqref{hoeq} is supposed to be 
\begin{equation}\label{symm}
Q = Q_1 + Q_2 + Q_3 + \cdots, \quad (Q_i\in\fr^i)
\end{equation}
where the leading term $Q_1$ is either $\Phi_m$ or $(\cD\Phi_m)$. 
For the latter, according to equation \eqref{def:com} we obtain the super commutator of $Q$ with itself, given by
\begin{equation*}
[Q,\ Q] = 2\Phi_{2m+1}+\cdots
\end{equation*}
which is a new symmetry of equation \eqref{hoeq}. Hence, without loss of generality, 
we only assume $Q$ is an $m$-th order nontrivial symmetry, which starts with $Q_1=\Phi_m(2\le m\neq n)$. 

To solve the symmetry condition $[K, Q]=0$, we break it up into
\begin{equation}\label{Condition}
\sum_{i=1}^r[ K_i,  Q_{r+1-i}]=0,
\end{equation}
which holds for all $r\geq 1$.
Clearly we have $ [ K_1,  Q_1]=0$ when $r=1$. The next equation to be solved 
is for $r=2$, that is,
\begin{equation}\label{quad}
 [K_1, Q_2]+ [K_2, Q_1]=0,
\end{equation}
which is trivially satisfied if equation (\ref{hoeq}) has no quadratic terms, i.e. $K_{2} = 0$. 
Let us concentrate on the case $K_{2}\neq 0$. 
We first rewrite equation \eqref{quad} in symbolic form using Corollary \ref{cor:1} or formula \eqref{Gfk}. 
Because $K_2\in \fr^{(1,1)}$, we assume that $\widehat{K_2}=\phi u F(\zeta_1,z_1)$. 
Then $\widehat{Q_2}$ is solved and formulated as
\begin{equation}\label{solqua}
 \widehat{Q_2}= \phi u \frac{F(\zeta_1,z_1) G_m^{(2)}(\zeta_1,z_1)} 
{G_n^{(2)}(\zeta_1,z_1)}=\phi u\frac{\Lambda(\zeta_1, z_1)}{\zeta_1 z_1 
(\zeta_1+z_1)}  G_m^{(2)}(\zeta_1,z_1),
\end{equation}
where $\lim_{\zeta_1+z_1\to 0} \Lambda(\zeta_1,z_1)$ exists. 

Taking $r=3$ in equation \eqref{Condition}, we have
\begin{equation}\label{cubcd}
 [K_1, Q_3]+[K_2, Q_2]+[K_3, Q_1]=0.
\end{equation}
To find $Q_3$, we have to compute $[K_2, Q_2]$, whose symbolic representation 
can be obtained by using Proposition \ref{prop:40}. Notice that there is a unique decomposition
\begin{equation*}
[K_2, Q_2] = P_{(1,2)} + P_{(3,0)}\quad\mbox{where $P_{(1,2)}\in\fr^{(1,2)}$ and $P_{(3,0)}\in\fr^{(3,0)}$}.
\end{equation*}
It follows from Theorem \ref{BeukId} that the nontrivial common factor occurs when both $n$ and $m$ are odd.
Thus we need to treat this case specially.
In the same way as we did for classification
of scalar homogeneous both commutative and non-commutative evolutionary equations,
we are able to show that $\widehat{P_{(1,2)}}$ is exactly divided by
$(\zeta_1+z_1)(\zeta_1+z_2)(z_1+z_2)$, while $\widehat{P_{(3,0)}}$ by
$(\zeta_1+\zeta_2)(\zeta_1+\zeta_3)(\zeta_2+\zeta_3)$ when both $m$ and $n$ are odd
and $ (\zeta_1+z_1)\big{|}\widehat{K_2}$ or $ \zeta_1z_1\big{|}\widehat{K_2}$.
Therefore, in this case, the existence of a symmetry is uniquely determined 
by the existence of its quadratic term. The proof is the same as we did in classical cases \cite{sand98,olver00}.
For completeness, we only present the statement.
\begin{The}
Suppose equation \eqref{hoeq} has a nonzero symmetry $Q$ of
the form \eqref{symm} with $m\geq 2$. Let
\begin{eqnarray}\label{ns}
S=\Phi_{m^\prime}+S_2+\cdots  \quad S_i\in \fr^i.
\end{eqnarray}
If a nonzero quadratic differential polynomial $S_2\in\fr^2$ satisfies $[\Phi_n, S_2]+ [K_2, \Phi_{m^\prime}]=0$ (cf. equation \eqref{quad} )
when $n$ and $m^\prime$ are both odd, then it uniquely determines a symmetry of the form \eqref{ns} for equation \eqref{hoeq}.
Moreover, these symmetries commute with each other, i.e. $[Q,\ S]=0$.
\end{The}
This theorem shows that the existence of one symmetry implies the existence of 
infinitely many symmetries as long as we know the existence of either 
quadratic or cubic terms of the symmetries for equation \eqref{hoeq}.

We have a very interesting observation implied by Theorem \ref{BeukId}. 
When $n$ and $m$ are both odd, the symbolic expression of $\widehat{Q_2}$ 
(cf. equation \eqref{solqua}) can be reformulated as
\begin{equation}\label{qua}
\frac{ \widehat{K_2}\ (\zeta_1^2+\zeta_1 z_1 + z_1^2)^{s-s'}\ H_m^{(2)}(\zeta_1,z_1)}
{H_n^{(2)}(\zeta_1,z_1)} ,
\end{equation}
where $s=\frac{m+3}{2}\pmod 3$ and $s'=\frac{n+3}{2}\pmod 3$. The quadratic term $Q_2$ 
is only well-defined if expression \eqref{qua} does be a polynomial, and further uniquely determines 
a symmetry $Q=\Phi_n+Q_2+\cdots$. As a necessary condition for this requirement, we deduce that
$\lambda \leq \frac{5}{2}+2\min(s,s')$.
The evolutionary equations defined by $Q$ has the same symmetries 
as equation \eqref{hoeq}. So instead of equation \eqref{hoeq} we
may consider the equation given by \(Q\). The lowest possible $m$ is  $2 s +3$
for $s=0,1,2$. Therefore we only need to consider $\lambda$-homogeneous 
equations with $\lambda \leq 6\frac{1}{2}$ of orders not greater than $7$.

A similar observation can be made for even $n>2$. Suppose we have found a
nontrivial symmetry with quadratic term
\begin{equation*}
\phi u \frac{F(\zeta_1, z_1)\ G_m^{(2)}}{\zeta_1 z_1 \,H_n^{(2)}} ,
\end{equation*}
which immediately implies $\lambda \leq \frac{3}{2}$. Then the quadratic term corresponding to $ 2
\phi u\frac{F(\zeta_1, z_1)}{ H_n^{(2)}} $ defines a symmetry starting with $\Phi_2$.
Therefore, we only need to find the symmetries of equations of order $2$ to 
get the
complete classification of symmetries of scalar $\lambda$-homogeneous 
equations (with $\lambda \leq \frac{3}{2}$) starting with an even linear term.

Finally, we analyze the case when equation (\ref{hoeq}) has no quadratic terms. 
If $K_i=0$ for $i=2,\cdots,j-1$, then we translate $[ K_{1},
 Q_j]+ [ K_{j},  Q_{1}]=0$ to symbolics, i.e.
\begin{equation*}\label{Srel}
 G_n^{(j)}\widehat{Q_j}=G_m^{(j)}\widehat{K _j}.
\end{equation*}
According to Theorem \ref{BeukId}, we know that there 
are
no symmetries for equation \eqref{hoeq} when \(j\geq 4\), or when \(j=3\) and \(n\) is 
even.
When \(j=3\) and \(n\) is odd, it can only have odd order symmetries. In this 
situation, 
one can remark that if the equation possesses symmetries for any $m$ then it 
must admit a symmetry of order $3$.
By now, we have proved the following statement:
\begin{The}\label{ReductionTheorem}
A nontrivial symmetry of a \( \lambda \)-homogeneous supersymmetric evolutionary
equation with $\lambda>0$ 
is part of a hierarchy starting at order \(2,\ 3,\ 5\) or \(7\).
\end{The}

Only equations with nonzero quadratic or cubic terms have nontrivial
symmetries. For each possible \(\lambda>0\), we must find a third order symmetry 
for a
second order equation, a fifth order symmetry for a third order equation, a 
seventh
order symmetry for a fifth order equation with quadratic terms, and the 
thirteenth
order symmetry for a seventh order equation with quadratic terms. 
Hence, there are finite possibilities for the order $n$ and $\lambda$, as illustrated by Table \ref{table:1}.
\begin{table}[!h]\centering
\begin{tabular}{c|cccccccccc}
   \hline
	$n$ &  \multicolumn{10}{|c}{$\lambda$}  \\
	\hline
	2 & 1/2, & & 3/2, & & & & & & & \\
	\hline
	3 & 1/2, & 1, & 3/2, & & 5/2, & & & & & \\
	\hline
	5 & 1/2, & 1, & 3/2, & 2, & 5/2, & & 7/2, & 9/2, & & \\
	\hline
	7 & 1/2, & 1, & 3/2, & 2, & 5/2, & 3, & 7/2, & 9/2, & 11/2, & 13/2 \\
	\hline
\end{tabular}	
\caption{All values of $n$ and $\lambda$} \label{table:1}
\end{table}

For every pair of $(n,\lambda)$, equation \eqref{susymodel} can be explicitly written out as supersymmetric differential 
polynomials in $\Phi$ with free coefficients. Let's take $(n,\lambda) = (3,1/2)$ 
as an example. The generic 3rd order $1/2$-homogeneous equation is given by
\begin{align}
\Phi_t =  \Phi_3 &+ c_{1}\Phi_2(\cD \Phi) + 
c_{2}\Phi_1(\cD \Phi_1) + c_{3}\Phi(\cD \Phi_2) + 
c_{4}\Phi_1(\cD \Phi)^2 \nonumber \\
& + c_{5}\Phi(\cD \Phi_1)(\cD \Phi) + c_{6}\Phi(\cD \Phi)^3 
. \label{eq3}
\end{align}
As we did for the supersymmetric KdV equation at the end of Section \ref{sec:4}, we compute 
its 5th order $1/2$-homogeneous symmetry $Q$ starting with the linear term $\Phi_5$. 
The existence of such symmetry leads to the conditions on the coefficients $c_i$ in \eqref{eq3},
which is a system of algebraic equations, whose solutions give us all integrable 
cases of equation \eqref{eq3}. 

For seventh order equations having thirteenth order symmetries, a fact indicated by their quadratic terms is that they all have 
fifth order symmetries. So there are no seventh order equations in our classification result. 

\section{Conclusions and discussions}
As an efficient tool to study the integrability, the symbolic method is successfully extended to the $N=1$ supersymmetric case in this paper. Any $N=1$ supersymmetric differential polynomial is associated to a unique multi-variable polynomial which is anti-symmetric with respect 
to symbols of fermionic variables and meanwhile symmetric with respect to symbols of bosonic variables. All operations on the ring of $N=1$ supersymmetric differential polynomials can be carried out in the symbolic way. On the basis of these results, we give a global classification for scalar $\lambda$-homogeneous $N=1$ supersymmetric evolutionary equations when $\lambda>0$, and identify eight leading equations having infinitely many higher order symmetries. Eight prototypes are respectively nontrivial supersymmetric integrable extensions of KdV, potential KdV, modified KdV, third order Burgers, Sawada-Kotera, fifth order KdV, Fordy-Gibbons and fifth order modified KdV equations. Two facts are noticed that there is no nontrivial supersymmetric integrable counterparts for the second order Burgers and Kaup-Kupershmidt equations, but there are two nonequivalent nontrivial supersymmetric integrable versions respectively for the fifth order KdV and modified KdV equations. Super symmetrisation both decreases and increases the number of integrable models.

Besides supersymmetric integrable evolutionary equations, non-evolutionary ones, especially equations of Camassa-Holm type, attract much attention in recent years. Various supersymmetric Camassa-Holm equations have been proposed through different approaches \cite{devsch01,luoxia13}, but their integrability has not been proved. 
In classical case, the symmetry approach in symbolic representation formulated in \cite{mn02} has been successfully applied to classification of Camassa-Holm type equations.
We expect that the formulation of symmetry approach in symbolic method developed in this paper can 
be applied to settle these problems.

\section*{Acknowledgment}
This work was done during KT's visit in the University of Kent, which is supported by 
the China Scholarship Council. KT would like to thank the School of Mathematics, Statistics \& Actuarial Science
for the hospitality. KT is also partially supported by the National Natural Science Foundation of China (NNSFC) (Grant Nos. 11271366, 11331008 and 11505284).

\appendix

\section{Scalar $\lambda$-homogeneous ($\lambda>0$) $N=1$ supersymmetric integrable equations: trivial cases}\label{app:a}
For completeness, trivial supersymmetric integrable equations are presented, and all of them are converted into classical integrable equation by the same potential transformation.
\begin{itemize}
\item[(\romannumeral1).] Trivial supersymmetric Burgers equation \cite{carstea02} ($\lambda = \frac{1}{2}$)
\begin{equation}\label{bsburgers}
\Phi_t = \Phi_{2} + 2\Phi_1(\cD \Phi) ,
\end{equation}
and it is rewritten in components as
\begin{equation*}
\left\{\begin{aligned}
v_t = v_{2} &+ 2v_1 v \\
\xi_t =\xi_{2} &+ 2\xi_1 v .
\end{aligned}\right.
\end{equation*}

Under the super Cole-Hopf transformation $\Phi=\cD \ln V$, where $V=V(x,\theta,t)$ is bosonic, equation \eqref{bsburgers} is linearized \cite{carstea02} to the heat equation $V_t = V_2$, and its infinitely many symmetries can be easily deduced from those of the heat equation via the Cole-Hopf transformation . 

\item[(\romannumeral2).] Trivial supersymmetric KdV equation \cite{becker93} ($\lambda = \frac{3}{2}$)
\begin{equation}\label{bskdv}
\Phi_t = \Phi_{3} + 6\Phi_1(\cD \Phi) ,
\end{equation}
and in components it is of the form
\begin{equation*}
\left\{\begin{aligned}
v_t = v_{3} &+ 6v_1v \\
\xi_t = \xi_{3} & +6\xi_1 v .
\end{aligned}\right.
\end{equation*}

Equation \eqref{bskdv} was claimed to be relevant to supersymmetric extensions of matrix models, or conformal field theories coupled to gravity \cite{becker93}.

\item[(\romannumeral3).] Trivial supersymmetric modified KdV equation ($\lambda = \frac{1}{2}$)
\begin{equation}\label{bsmkdv}
\Phi_t = \Phi_{3} + 6\Phi_1(\cD \Phi)^2 ,
\end{equation}
and its component form is given by
\begin{equation*}
\left\{\begin{aligned}
v_t = v_{3} &+ 6v_1v^2 \\
\xi_t =\xi_{3} &+ 6\xi_1v^2 .
\end{aligned}\right.
\end{equation*}

\item[(\romannumeral4).] Trivial supersymmetric Sawada-Kotera equation \cite{tian10} ($\lambda = \frac{3}{2}$)
\begin{equation}\label{bssk}
\Phi_t  = \Phi_{5} + 5\Phi_{3}(\cD \Phi) + 5\Phi_1(\cD \Phi_{2}) + 5\Phi_1(\cD \Phi)^2 ,
\end{equation}
and in components it is rewritten as
\begin{equation*}
\left\{\begin{aligned}
v_t = v_{5} &+ 5v_{3}v + 5v_{2}v_x + 5v_1v^2 \\
\xi_t = \xi_{5} &+ 5\xi_{3}v + 5\xi_1 v_{2x} + 5\xi_1 v^2 .
\end{aligned}\right.
\end{equation*}

\item[(\romannumeral5).] Trivial supersymmetric Kaup-Kupershmidt equation \cite{tian10} ($\lambda = \frac{3}{2}$)
\begin{equation}\label{bskk}
\Phi_t = \Phi_{5} + 10\Phi_{3}(\cD \Phi) + 15\Phi_{2}(\cD \Phi_1) + 10\Phi_1(\cD \Phi_{2}) + 20\Phi_1(\cD \Phi)^2 ,
\end{equation}
and its component form is 
\begin{equation*}
\left\{\begin{aligned}
v_t = v_{5} &+ 10v_{3}v + 25v_{2}v_1 + 20v_1v^2 \\
\xi_t = \xi_{5} &+ 10\xi_{3}v + 15\xi_{2}v_1 + 10\xi_{1}v_{2} + 20\xi_1v^2 .
\end{aligned}\right.
\end{equation*}

\item[(\romannumeral6).] Trivial supersymmetric Fordy-Gibbons equation ($\lambda = \frac{1}{2}$)
\begin{align}
\Phi_t = \Phi_{5} &- 5\Phi_{3}(\cD \Phi_1) - 5\Phi_{3}(\cD \Phi)^2 - 5\Phi_{2}(\cD \Phi_{2}) - 10\Phi_{2}(\cD \Phi_1)(\cD \Phi) \nonumber \\
& - 10\Phi_1(\cD \Phi_{2})(\cD \Phi) - 5\Phi_1(\cD \Phi_1)^2 + 5\Phi_1(\cD \Phi)^4 , \label{bsfg}
\end{align}
and it is rewritten in components as 
\begin{equation*}
\left\{\begin{aligned}
v_t =& v_{5} - 5v_{3}v_1 - 5v_{3}v^2 - 5v_{2}^2 - 20v_{2}v_1v -  5v_1^3 + 5v_1v^4\\
\xi_t =& \xi_{5} - 5\xi_{3}(v_1 + v^2) - 5\xi_{2}(v_{2} + 2v_1v) - 5\xi_1(2v_{2}v + v_1^2 - v^4) .
\end{aligned}\right.
\end{equation*}
\end{itemize}

As a common remark to six trivial extensions, their triviality can be alternatively understood by introducing a potential, i.e. $\Phi = (\cD W)$, where $W=W(x,\theta,t)$ is bosonic. Then it is straightforward to show equations \eqref{bsburgers}--\eqref{bsfg} are respectively converted into
\begin{align*}
\text{Potential Burgers equation: }& W_t = W_{2} + W_1^2 ,\\
\text{Potential KdV equation: }& W_t = W_{3} + 3W_1^2 ,\\
\text{Potential modified KdV equation: }& W_t = W_{3} + 2W_1^3 ,\\
\text{Potential Sawada-Kotera equation: }& W_t = W_{5} + 5W_{3}W_1 + \frac{5}{3}W_1^3 ,\\
\text{Potential Kaup-Kupershmidt equation: }& W_t = W_{5} + 10W_{3}W_1 + \frac{15}{2}W_{2}^2 + \frac{20}{3}W_1^3 ,\\
\text{Potential Fordy-Gibbons equation: }& W_t = W_{5} - 5W_{3}W_{2} - 5W_{3}W_1^2 - 5W_{2}^2W_1 + W_1^5 .
\end{align*}
Through the potential transformation, recursion operators and/or master symmetries of equations \eqref{bsburgers}--\eqref{bsfg} are easily constructed from the relevant results about these potential equations \cite{wang02}  (or references therein), and hence will not be presented here.

\end{document}